\documentclass[11pt]{article}

\usepackage[letterpaper,margin=1.00in]{geometry}
\usepackage{amsmath, amssymb, amsthm, thmtools, amsfonts}
\usepackage{bbm}

\usepackage{ifthen}

\usepackage{tikz}
\usetikzlibrary{positioning,decorations.pathreplacing}

\usepackage{cite}
\usepackage{appendix}
\usepackage{graphicx}
\usepackage{color}
\usepackage{algorithm}
\usepackage[noend]{algpseudocode}
\usepackage{epstopdf}
\usepackage{tikz}
\usetikzlibrary{arrows}
\usetikzlibrary{calc}
\usetikzlibrary{fadings}

\usepackage[textsize=tiny]{todonotes}

\usepackage{framed}
\usepackage[framemethod=tikz]{mdframed}
\usepackage[bottom]{footmisc}
\usepackage[shortlabels]{enumitem}
\setitemize{noitemsep,topsep=3pt,parsep=3pt,partopsep=3pt}
\usepackage[font=small]{caption}
\usepackage{xspace}

\usepackage{mathtools}

\newtheorem{theorem}{Theorem}[section]
\newtheorem{lemma}[theorem]{Lemma}
\newtheorem{meta-theorem}[theorem]{Meta-Theorem}

\newtheorem{corollary}[theorem]{Corollary}

\newtheorem{definition}[theorem]{Definition}

\definecolor{darkgreen}{rgb}{0,0.5,0}
\usepackage{hyperref}
\hypersetup{
    unicode=false,          % non-Latin characters in Acrobat’s bookmarks
    colorlinks=true,        % false: boxed links; true: colored links
    linkcolor=red,          % color of internal links (change box color with linkbordercolor)
    citecolor=darkgreen,        % color of links to bibliography
    filecolor=magenta,      % color of file links
    urlcolor=cyan           % color of external links
}
\usepackage[capitalize, nameinlink]{cleveref}
\crefname{theorem}{Theorem}{Theorems}
\Crefname{lemma}{Lemma}{Lemmas}
\Crefname{observation}{Observation}{Observations}
\Crefname{equation}{}{}
\algnewcommand\algorithmicswitch{\textbf{switch}}
\algnewcommand\algorithmiccase{\textbf{case}}

\algdef{SE}[SWITCH]{Switch}{EndSwitch}[1]{\algorithmicswitch\ #1\ \algorithmicdo}{\algorithmicend\ \algorithmicswitch}%
\algdef{SE}[CASE]{Case}{EndCase}[1]{\algorithmiccase\ #1}{\algorithmicend\ \algorithmiccase}%
\algtext*{EndSwitch}%
\algtext*{EndCase}%
\newcommand{\eps}{\varepsilon}

\newcommand{\poly}{\operatorname{\text{{\rm poly}}}}

\newcommand{\Var}{\text{Var}}

\DeclareMathOperator{\E}{\mathbb{E}}

\renewcommand{\paragraph}[1]{\vspace{0.15cm}\noindent {\bf #1}:}

\usepackage{setspace}

\usepackage{mathtools}
\DeclarePairedDelimiter{\abs}{\lvert}{\rvert}

\makeatletter
\let\oldabs\abs
\def\abs{\@ifstar{\oldabs}{\oldabs*}}
\makeatother

\newcommand{\FullOrShort}{full}

\ifthenelse{\equal{\FullOrShort}{full}}{
	
  \newcommand{\fullOnly}[1]{#1}
  \newcommand{\shortOnly}[1]{}
  
  }{

    \newcommand{\fullOnly}[1]{}
    \newcommand{\IncludePictures}[1]{}
   
  }

\begin{document}

\date{}

\title{Tight Analysis of Parallel Randomized Greedy MIS}

\author{
	 Manuela Fischer\\
  \small ETH Zurich \\
  \small manuela.fischer@inf.ethz.ch
		\and
 Andreas Noever\\
  \small ETH Zurich\\
  \small anoever@inf.ethz.ch
 }

\maketitle

\setcounter{page}{0}
\thispagestyle{empty}

\begin{abstract}
We provide a tight analysis which settles the round complexity of the well-studied \emph{parallel randomized greedy MIS} algorithm, thus answering the main open question of Blelloch, Fineman, and Shun [SPAA'12]. 

The parallel/distributed randomized greedy Maximal Independent Set (MIS) algorithm works as follows. An order of the vertices is chosen uniformly at random. Then, in each round, all vertices that appear before their neighbors in the order are added to the independent set and removed from the graph along with their neighbors. The main question of interest is the number of rounds it takes until the graph is empty. This algorithm has been studied since 1987, initiated by Coppersmith, Raghavan, and Tompa [FOCS'87], and the previously best known bounds were $O(\log n)$ rounds in expectation for Erd\H{o}s-R\'{e}nyi random graphs by Calkin and Frieze [Random Struc. \& Alg. '90] and $O(\log^2 n)$ rounds with high probability for general graphs by Blelloch, Fineman, and Shun [SPAA'12].

We prove a high probability upper bound of $O(\log n)$ on the round complexity of this algorithm in general graphs, and that this bound is tight. This also shows that parallel randomized greedy MIS is as fast as the celebrated algorithm of Luby [STOC'85, JALG'86].
\end{abstract}

\newpage

\section{Introduction and Related Work}
The Maximal Independent Set (MIS) problem plays a central role in parallel and distributed computing \cite{valiant1983parallel,cook1983overview}, and has---due to its many applications in symmetry breaking \cite{luby1985simple}---been extensively studied for more than three decades \cite{karp1985fast,luby1985simple,alon1986fast,goldberg1986parallel,goldberg1987parallel,coppersmith1989parallei,goldberg1989constructing,goldberg1989new,calkin1990probabilistic,linial1992locality,blelloch2012greedy,barenboim2012locality,Ghaffari-MIS}. 
We refer to \cite{DBLP:journals/jacm/BarenboimEPS16} for a thorough review of the state of the art. 

One strikingly simple algorithm addressing this problem is the \emph{parallel/distributed randomized greedy MIS} algorithm, which works as follows. An order of the vertices is chosen uniformly at random. Then, in each round, all \emph{local minima}---i.e., all vertices that appear before their neighbors in the order---are added to the independent set and removed from the graph along with their neighbors. The main question of interest is the number of rounds it takes until the graph is empty.

This algorithm is particularly easy to implement and requires only a small amount of communication. Indeed, a vertex only needs to inform its neighbors about its position in the random order and then, in the round of its removal from the graph, about its decision (whether to join the MIS). For practical implementation-related details, we refer to \cite{blelloch2012greedy}. 

Another nice property of this algorithm is that---once an order is fixed---it always yields the so-called \emph{lexicographically first MIS}, i.e., the same as the \emph{sequential greedy MIS} algorithm that goes through the vertices in this order one by one and adds a vertex to the MIS if none of its neighbors has been added in a previous step. Such determinism can be an important feature of parallel algorithms \cite{bocchino2009parallel,blelloch2012internally}. 

These practical advantages are mainly owed to the fact that the same random order is used throughout. This, however, comes with the drawback of complicating the analysis significantly, due to the lack of independence among different rounds. Indeed, while for \emph{Luby's algorithm}---which is the same algorithm but with regenerated random order in each iteration---the round complexity was established at $O(\log n)$ for general $n$-node graphs more than 30 years ago \cite{luby1985simple,alon1986fast}, no similar result is known for parallel randomized greedy MIS.   
For Erd\H{o}s-R\'enyi random graphs, Calkin and Frieze \cite{calkin1990probabilistic} could prove an (in expectation) upper bound of $O(\log n)$, resolving the conjecture of Coppersmith, Raghavan, and Tompa \cite{coppersmith1989parallei} who themselves arrived at $O\left(\frac{\log^2 n}{ \log \log n}\right)$. A matching lower bound of $\Omega(\log n)$ was proven by Calkin, Frieze and Ku\v{c}era\cite{calkin1992probabilisticlower} two years later.

For general graphs, Blelloch, Fineman, and Shun \cite{blelloch2012greedy} proved 5 years ago that w.h.p.\footnote{As standard, \emph{with high probability}, abbreviated as w.h.p., indicates a probability at least $1-\frac{1}{n^c}$, for any desirably large constant $c\geq 2$.} $O(\log^2 n)$ rounds are enough.
The authors stated as one of their main open questions  
\begin{center}
\begin{minipage}{0.95\linewidth}
\vspace{-10pt}
\begin{mdframed}[hidealllines=true, backgroundcolor=gray!00]
\center{
\emph{``whether the dependence length [...] can be improved to $O(\log n)$",}
}
\vspace{-10pt}
\end{mdframed}
\end{minipage}
\end{center}
thus whether the analysis of parallel greedy MIS's round complexity can be improved to $O(\log n)$. 
\subsection{Our Results} 
We show that the parallel greedy MIS's round complexity is indeed in $O(\log n)$. This in particular also shows that parallel randomized greedy MIS is as fast as the celebrated algorithm of Luby, confirming a widespread belief.
 
\begin{theorem}\label{thm:MIS}
The parallel/distributed randomized greedy MIS algorithm terminates in $O(\log n)$ rounds on any $n$-node graph with high probability.  
\end{theorem}

The result of Calkin, Frieze and Ku\v{c}era\cite{calkin1992probabilisticlower} proves that this is asymptotically best possible and we also provide an alternative short proof of the lower bound in \Cref{appendix:lower}. In \Cref{appendix:implications}, we present implications of \Cref{thm:MIS} for maximal matching and $(\Delta+1)$-vertex-coloring as well as the correlation clustering problem. 

\subsection{Overview of Our Method}
It is a well-known fact that the removal of all local minima along with their neighbors from a graph for a random order in expectation leads to a constant factor decrease in the total number of edges (see, e.g., \cite{YvesMIS} for a simple proof). When---as it is the case in Luby's algorithm---the random order is regenerated in every iteration, repeated application of this argument directly yields an upper bound of $O(\log n)$ on the round complexity. However, if the order is kept fixed between rounds, then the order of the remaining vertices is no longer uniformly distributed. 

To overcome this problem of dependencies among different iterations, Blelloch, Fineman, and Shun \cite{blelloch2012greedy}---inspired by an approach of \cite{coppersmith1989parallei} and \cite{calkin1990probabilistic}---divide the algorithm into several phases. In each phase they only expose a \emph{prefix} of the remaining order and run the parallel algorithm on these vertices only (whilst still deleting a vertex in the suffix if it is adjacent to a vertex added to the MIS).
This way, in each phase the order among the unprocessed (but possibly already deleted) vertices in the suffix remains random, leading to a sequence of ``independent" problems. 

They exploit this independence to argue that after having processed a prefix of length $\Omega(t \cdot \log n)$, the maximum remaining degree (among the not yet deleted vertices) in the suffix w.h.p.\ is $d=\frac{n}{t}$. This is because in every step\footnote{For the sake of analysis, we think of the prefix being processed sequentially. This does not change the outcome.}
in which a vertex $v$ has more than $d$ neighbors, the probability that one of these is chosen to be exposed next (which causes the deletion of $v$) is at least $\frac{d}{n}$. In the end of the phase, the probability of $v$ not being deleted is at most $\left(1-\frac{d}{n}\right)^{\Omega(t \cdot \log n)}=\frac{1}{\poly n}$. A union bound over all vertices concludes the argument. By doubling the parameter $t$ after each phase they ensure that after $O(\log n)$ phases the whole graph has been processed.
Inside each phase, they use the maximum degree to bound the round complexity by the length of a longest monotonically increasing path\footnote{A monotonically increasing path with respect to an order is a path along which the order ranks are increasing.}, which is $O(\log n)$. 

The main shortcoming of Blelloch et al.'s approach is that it relies heavily on the property that in each phase the remaining degree of \emph{all} vertices \emph{with high probability} falls below a certain value. This imposed union bound unavoidably stretches each prefix by a factor of $\log n$. 
We will circumvent this problem using the following core ideas.

%\paragraph{Main Ideas}

\begin{enumerate}[(i),topsep=0pt]
\item
Instead of bounding the degree of \emph{all} vertices in the graph, we will consider a fixed set of $O(\log n)$ \emph{positions}, that is, indices in $\{1,\dotsc,n\}$, and only analyze the degree of vertices assigned to these positions in the random order. 
\item
Instead of using one prefix for all these vertices \emph{simultaneously}, we will essentially have \emph{one distinct ``prefix"} for each position, that is, one distinct set of vertices which we will use to argue about the degree drop of the vertex assigned to this position. This will preserve independence among positions, and hence spare us the need of a union bound.  
\item Instead of bounding the probability of a long \emph{monotonically increasing path} of vertices based on the \emph{with high probability upper bound} on the degree, we will restrict our attention to the much stronger concept of so-called \emph{dependency paths}\footnote{Note that this is not the same as \emph{dependence path} in \cite{blelloch2012greedy}.}.
Roughly speaking, a dependency path is a monotonically increasing path which alternates between vertices in the MIS and vertices not in the MIS, and the predecessor of a vertex $v$ not in the MIS is the first vertex that knocked $v$ out. These additional properties of dependency paths allow us to execute a more nuanced analysis of their occurrence probability. In particular, we will not need to argue that the degrees of our vertices drop according to some process with high probability. It suffices to show that for any degree of this vertex its chances of being part of a dependency path are low enough.
\end{enumerate}

\paragraph{Proof Outline} Summarized, our method---comprising all the aforementioned ideas---can be outlined as follows. We will show that the parallel round complexity is bounded by the largest dependency path length. It is then enough to show that w.h.p.\ there cannot be a dependency path of length $L=\Omega(\log n)$. In a first step, we will analyze the probability that a fixed set $P$ of positions forms a dependency path. To this end, we assign each position a position set (which will play the role of a ``prefix'' for this position's vertex and thus serve the purpose of controlling its degree) of a certain size and at a certain place, both carefully chosen depending on the position. Then we will argue for each position\footnote{In fact, to get a strong enough bound, we will have to argue not only about one position but about two positions simultaneously.} that its probability of being part of and continuing the dependency path is not too high, based on the randomness of its associated position set. Being careful about only exposing positions that have not already been exposed for other positions, we will be able to combine these probabilities to obtain a  
bound on the probability that $P$ forms a dependency path. Finally, we union bound over all choices for $P$.

\subsection{Notation}
We use $[n]:=\{1, \dotsc, n\}$ and $[x,y]:=\{x, \dotsc, y\}$. For two sets $X,Y \subseteq \mathbb{N}$, we write $X < Y$ if $\max X < \min Y$. We use $X[i]$ for the $i^{\text{th}}$ element in the set $X$ (we think of all sets as ordered) and $X[I]$ the (ordered) set of the elements in $X$ at positions $I\subseteq \left[|X|\right]$. 
For an order $\pi \colon [n] \rightarrow V$, we say that vertex $v\in V$ has position $i$ if $\pi(i)=v$,
use $\pi(I):=\bigcup_{i \in I} \pi(i)$, and write $\pi(I)=\pi'(I)$  when $\pi(i)=\pi'(i)$ for all $i \in I$. 
We say that we \emph{expose} a position $i$ when we fix the vertex $\pi(i)$ in a random order $\pi$. Moreover, we say a vertex $v$ is exposed if there is a position $i$ such that $i$ is exposed and $\pi(i)=v$. If a set $I$ of positions is already exposed, this means that the considered probabilities are all conditioned on $\pi(I)$. We use subscript $I$ to indicate that the probability is over the randomness in the positions $I$. For a graph $G=(V,E)$, we use $E(X,Y)$ to denote the set of edges in $E$ between $X$ and $Y$, for $X,Y \subseteq V$, and write $e(X,Y)=|E(X,Y)|$. Moreover, we let $N(v):=\{u \in V \colon \{u,v\}\in E\}$ denote the neighborhood of vertex $v\in V$.

\section{Proof of \texorpdfstring{\Cref{thm:MIS}}{Theorem~\ref{thm:MIS}}: Upper Bound}\label{section:upper}
\subsection{Framework}\label{framework}

We introduce the concept of \emph{dependency paths},  and show that this notion is closely related to the round complexity of the parallel greedy MIS algorithm. %and will give us an easier way to reason about them. 

\paragraph{Dependency Path}
For a fixed permutation $\pi \colon [n] \rightarrow V$, let $V^*\subseteq V$ denote the MIS generated by the (sequential) greedy algorithm that processes the vertices in the order $(\pi(1), \dotsc, \pi(n))$. For every vertex $v$ not in $V^*$, we use $\text{inhib}(v)$ to denote the neighbor of $v$ in $V^*$ of minimum position, that is, setting $\text{inhib}(v):=\arg \min \{ \pi^{-1}(u) \colon u \in N(v) \cap V^*\}$, and call it $v$'s \emph{inhibitor}. 

\begin{definition}\label{def:dependencypath} %[Dependency Path]
A sequence $1\leq p_1 < \dotsb < p_{2l+1} \leq n$ of positions forms a \emph{dependency path} of length $2l+1$ for $l \geq 0$ if
\begin{enumerate}[(i)]
\item  $(\pi(p_1),  \dotsc, \pi(p_{2l+1}))$ is a path in $G$,
\item $\{\pi(p_k)\mid k \text{ is odd}\}\subseteq V^*$,
\item $\{\pi(p_k) \mid k \text{ is even}\} \subseteq V \setminus V^*$,
\item $\pi(p_{k-1})=\text{inhib}(\pi(p_k))$ for even $k$.
\end{enumerate}
We write $p_1 \sim \dotsb \sim p_{2l+1}$.
\end{definition}

\paragraph{Connection to Parallel Algorithm} In the following, we establish a connection between the round complexity of the parallel algorithm and the \emph{dependency length}, defined as the length of the longest dependency path in a graph.

\begin{lemma}\label{lemma:reductionToLength}
If the dependency length is $2l+1$, the parallel algorithm takes at most $l+1$ rounds.
\end{lemma}
\begin{proof}
Consider a slowed-down version of the parallel algorithm in which the deletion of a vertex $v\notin V^*$ is delayed until the round in which $\text{inhib}(v)$ enters the independent set. That is, even if a neighbor $u$ of $v$ enters the independent set, $v$ is not deleted from the graph unless $u =\text{inhib}(v)$. This algorithm takes at least as many rounds as the original parallel algorithm, as the time in which a vertex is processed can only be delayed.  Furthermore the slowed-down version produces the same independent set as in the original algorithm.

We show by induction that every vertex entering the independent set in round $i$ in this modified parallel algorithm must be the last vertex of a dependency path of length $2(i-1)+1$. 
The base case $i=1$ is immediate.
If a vertex $w$ enters $V^*$ in round $i+1$, then there is a neighbor $v$ of $w$ with $\pi(v) < \pi(w)$ that was deleted from the graph (but not added to the MIS) in round $i$, as otherwise $w$ could have been added to $V^*$ in an earlier round. This means that $\text{inhib}(v)$ was added to $V^*$ in round $i$. By the induction hypothesis, $\text{inhib}(v)$ is the last vertex of a dependency path of length $2(i-1)+1$, and it is easy to check that $v$ and $w$ extend this dependency path.
\end{proof}

\subsection{Proof Outline}\label{outline}
The number of rounds taken by the randomized parallel greedy MIS algorithm in the beginning and in the end---that is, for vertices with positions in $[1, \Theta(\log n)]$ and $[\beta n, n]$, for a constant $\beta \in (0,1)$ which is given by \Cref{lemma:continuingpath} below---can be handled easily, as we will discuss in the proof of \Cref{thm:MIS} in \Cref{section:unionboundproof}. We thus focus on the technically more interesting range of positions in $[\Theta(\log n), \beta n]$ here. By \Cref{lemma:reductionToLength}, we know that the dependency length constitutes an upper bound on the round complexity of the parallel greedy MIS algorithm. It is thus enough to show the following.

\begin{theorem}\label{lemma:dependencyLength}
W.h.p.\ there is no dependency path of length $L=\Omega(\log n)$ in the interval $[ \beta  n]$.
\end{theorem} 
One main part of this theorem's proof, which appears in \Cref{section:unionboundproof}, is to argue that the probability that a fixed set $\{p_1, \dotsc, p_L\}$ of positions form a dependency path is low. This, in turn, is proved by bounding the probability that---when already having exposed $p_{k-1}$---the positions $p_k$ and $p_{k+1}$ continue a dependency path, thus $p_{k-1}\sim p_k \sim p_{k+1}$, for all segments $(p_{k-1}, p_k, p_{k+1})$ for even $k\in [L]$. By being careful about dependencies among segments, thus in particular about which randomness to expose when, these probabilities for the segments then can be combined to a bound for the probability of $p_1\sim \dotsb \sim p_L$. We will achieve this by assigning disjoint position sets $P_i$ to positions $p_i$ and exposing, roughly speaking, only $P_{k-1}, P_k$, and $\{p_{k}, p_{k+1}\}$ when analyzing the probability of the segment $(p_{k-1}, p_k, p_{k+1})$. More formally, this reads as follows. 
\begin{lemma}\label{lemma:continuingpath}
There exist absolute constants $\beta, \eps \in (0,0.01)$ such that the following holds. 
Fix three positions $p_1, p_2, p_3 \in [\beta n]$ and two disjoint sets $P_1,P_2\subseteq [ \beta n]$  of equal size $l\coloneqq |P_1|=|P_2| \geq 10000$ which satisfy $P_1 < P_2 < p_1 < p_2 < p_3$, and let $t_2:=\max P_2$. Consider $S \subseteq [ \beta n]\setminus (P_1 \cup P_2 \cup \{p_2,p_3\})$ with $p_1 \in S$. Suppose that the positions in $S$ have already been exposed. Then the probability that $p_1, p_2, p_3$ can still form a dependency path when additionally exposing $S_2\setminus S$ for $S_2\coloneqq S \cup [t_2]\cup \{p_2,p_3\}$ is $\Pr_{S_2}[ p_1\sim p_2 \sim p_3 \mid \pi\left(S\right)] \leq \frac{1-\eps}{(e\cdot l)^2}$.
\end{lemma}
The proof of this lemma is deferred to \Cref{sec:continuingpath}.
 
\subsection{Proofs of \texorpdfstring{\Cref{lemma:dependencyLength,thm:MIS}}{Theorems~\ref{lemma:dependencyLength} and \ref{thm:MIS}}}\label{section:unionboundproof}

We will bound the probability that there exists a dependecy path of length $\Omega(\log n)$ in the interval $[A\log n,  \beta n]$ for some large $A$, thereby proving \Cref{lemma:dependencyLength}.  

\begin{proof}
We first upper bound the probability that a fixed set $P$ of $L=\Theta(\log n)$ (for odd $L$) positions forms a dependency path by iteratively applying \Cref{lemma:continuingpath}. We then take a union bound over all possible choices for $P$. 

\paragraph{Probability for Fixed Positions}
Let $I=[1,\dots,n]$ and let $A$ denote a large constant which we will fix later. 
For a fixed set $P=\left\{p_1,\dots,p_L\right\}\subseteq [A\log n,n]$ of $L=\sqrt{A}\log n$ positions, we will choose position sets $P_k$ to associate with each $p_k$ which satisfy the conditions
\begin{enumerate}[(1)]
\item $P_1 < \dotsb < P_{L}$,
\item $ P_k<p_k$ for all $k\in [L]$, and additionally 
\item $P_k<  p_{k-1}$ and $\abs{P_k} = \abs{P_{k-1}}$ for all even $k\in [L]$.
\end{enumerate}
This will ensure the applicability of \Cref{lemma:continuingpath} to $(p_{k-1}, p_k, p_{k+1})$ for even $k \in [L]$.

Break $[n]$ into exponentially growing sub-intervals $I_i\coloneqq\left[(1+\alpha)^{i}, (1+\alpha)^{i+1}\right)$ for some small enough constant $0<\alpha  < 1$.
Let $s_i:=|I_i \cap (P\setminus \{p_L\})|$ denote the number of positions in $P \setminus \{p_L\}$ that intersect $I_i$. Observe that $s_i=0$ for all $i+1 \leq \log_{1+\alpha} (A \log n)$.
For convenience, we use $p_{i,j}$ to refer to $p_k$ if $p_k$ is the $j^{\text{th}}$ position among $I_i\cap P$.

If $s_i>0$ then we split $I_{i-1}\setminus P$ into $s_i+1$ position sets $I_{i-1,1}< \dotsc< I_{i-1,s_i+1}$, each of size 
\[
	l_i:=\left\lfloor\frac{\abs{I_{i-1}\setminus P}}{s_i+1}\right\rfloor
		=\left\lfloor\frac{\lfloor\alpha(1+\alpha)^{i-1}\rfloor-\abs{P\cap I_{i-1}}}{s_i+1}\right\rfloor.
	%	=\frac{\alpha(1+\alpha)^{i-1}}{s_i+1}\left(1-\frac{\abs{P\cap I_{i-1}}}{\alpha\left(1+\alpha\right)^{i-1}}\right).
%		\geq\frac{\alpha(1+\alpha)^{i-1}}{s_i+1}\left(1-\alpha\right),
\]
Note that these subsets are not necessarily contiguous and do not contain any of the positions $p_k\in P$. We claim that $l_i \geq \alpha(1+\alpha)^{i-2}/(s_i+1)\geq 10000$. Indeed for $s_i$ to be non-zero $i$ must satisfy $(1+\alpha)^{i+1} \geq A \log n = \sqrt{A}\abs{P}$. Thus the first inequality holds for $A(\alpha)$ large enough. Secondly $s_i \leq \abs{P}$ thus the second inequality holds for $A$ large enough as well.

Next we assign each position $p_k=p_{i,j}\in P$ a position set 
\[
P_k=I(p_k)=I(p_{i,j}) := \begin{cases}
	I_{i-1, j}, &\text{if $k$ is odd or $j\neq 1$,} \\
	I_{i-h_i-1,s_{(i-h_{i})}+1}, &\text{if $k$ is even, $j= 1$},
\end{cases}
\]
where $h_i:=i-\max\{j\mid j < i \colon s_j\neq 0\}$ 
is the distance $i-j$ of $I_i$ to the next smaller interval $I_j$ which contains at least one position from $P$.\footnote{Note that, since we are interested in $h_i$ only for positions $p_k=p_{i,1}$, where $k$ is even. Thus in particular not for $p_1$, there will actually always be a $j < i$ with $s_j\neq 0$.}  In words, we assign the $j^{\text{th}}$ position $p_{i,j}$ in the interval $I_i$ the $j^{\text{th}}$ position set $I_{i-1,j}$ in the previous interval $I_{i-1}$, except for positions $p_k$ for which $k$ is even and $p_k$ is the first position in some interval $I_i$. Since in that case $p_{k-1}$ is in the interval $I_{i-h_i}$, and thus is assigned a position set $I_{i-h_i-1,s_{(i-h_i)}}$, we have to assign $p_{k}$ the position set $I_{i-h_i-1,s_{(i-h_i)}+1}$ in order to not violate condition (3).

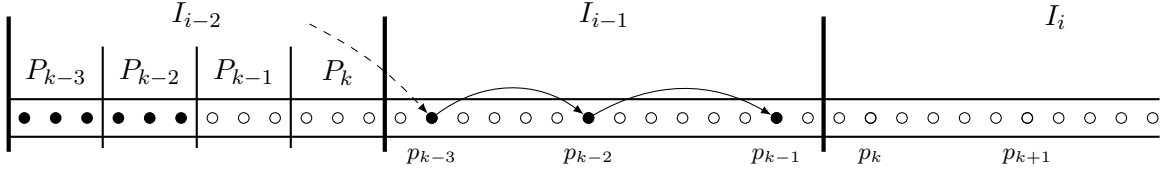
\begin{figure}
\begin{tikzpicture}[>=latex]
	\tikzstyle{point}=[draw,circle,black,minimum size=4pt,inner sep=0pt]
	\useasboundingbox (-0.1,-0.4) rectangle (15.5,1.9);
      %\draw (-0.1,-0.4) rectangle (15.5,1.9);

	\node (L1) at (0,0) {};
	\node (L2) at (5,0) {};
	\node (L3) at ($(10+2*5/12,0)$) {};
	\node (L4) at (17,0) {}; % not shown, positioning of $I_i

	\coordinate (L11) at ($0.75*(L1)+0.25*(L2)$);
	\coordinate (L12) at ($0.5*(L1)+0.5*(L2)$);
	\coordinate (L13) at ($0.25*(L1)+0.75*(L2)$);;

	\coordinate (L1A) at ($0.5*(L1)+0.5*(L11)$);
	\coordinate (L1B) at ($0.5*(L11)+0.5*(L12)$);
	\coordinate (L1C) at ($0.5*(L12)+0.5*(L13)$);
	\coordinate (L1D) at ($0.5*(L13)+0.5*(L2)$);

	\coordinate (L2M) at ($0.5*(L2)+0.5*(L3)$);
	\coordinate (L2ML) at ($0.5*(L2)+0.5*(L2M)$);
	\coordinate (L2MR) at ($0.5*(L2M)+0.5*(L3)$);
	\coordinate (L3M) at ($0.5*(L3)+0.5*(L4)$);
	\coordinate (L3ML) at ($0.5*(L3)+0.5*(L3M)$);
	\coordinate (L3MR) at ($0.5*(L3M)+0.5*(L4)$);

	% Horizontal lines
	\draw [thick] ($(L1)$) -- ($(15.3,0)$);
	\draw [thick] ($(L1)+(0,0.5)$) -- ($(15.3,0)+(0,0.5)$);

	% Vertical lines
	\draw [ultra thick] ($(L1)-(0,0.2)$) -- ($(L1)+(0,1.6)$);
	\draw [thick] ($(L11)-(0,0.15)$) -- ($(L11)+(0,1.2)$);
	\draw [thick] ($(L12)-(0,0.15)$) -- ($(L12)+(0,1.2)$);
	\draw [thick] ($(L13)-(0,0.15)$) -- ($(L13)+(0,1.2)$);
	\draw [ultra thick] ($(L2)-(0,0.2)$) -- ($(L2)+(0,1.6)$);
	\draw [ultra thick] ($(L3)-(0,0.2)$) -- ($(L3)+(0,1.6)$);

	% points
	 \foreach \x in {0,...,5}
          \node (P\x) at ($({(\x+0.5)*5/12}, 0.25)$) [point,fill] {};
	 \foreach \x in {6,...,36}
            \node  (P\x) at ($({(\x+0.5)*5/12}, 0.25)$) [point] {};

      \node at (P13) [point,fill,label={[label distance=0.2cm]below:\footnotesize{$ p_{k-3}$}}] {};
      \node at (P18) [point,fill,label={[label distance=0.2cm]below:\footnotesize{$ p_{k-2}$}}] {};
      \node at (P24) [point,fill,label={[label distance=0.2cm]below:\footnotesize{$ p_{k-1}$}}] {};

      \node at (P27) [point,label={[label distance=0.2cm]below:\footnotesize{$ p_{k}$}}] {};
      \node at (P32) [point,label={[label distance=0.2cm]below:\footnotesize{$ p_{k+1}$}}] {};

	% arrows
	\path[dashed,->] (4,1.5) edge [bend left=10]  (P13);
	\path[->] (P13) edge [bend left=37] (P18);
	\path[->] (P18) edge [bend left] (P24);

	% sets P_
      \node at ($(L1A)+(0,0.85)$) {$P_{k-3}$};
      \node at ($(L1B)+(0,0.85)$) {$P_{k-2}$};
      \node at ($(L1C)+(0,0.85)$) {$P_{k-1}$};
      \node at ($(L1D)+(0,0.85)$) {$P_{k}$};

	% sets I_
      \node [above=1.3cm of L12] {$I_{i-2}$};
      \node [above=1.3cm of L2M] {$I_{i-1}$};
      \node [above=1.3cm of L3M] {$I_{i}$};
\end{tikzpicture}

\caption{The situation before applying \Cref{lemma:continuingpath} to $(p_{k-1},p_k,p_{k+1})$. The set $I_{i-2}\setminus P$ has been split into $s_{i-1}+1=4$ parts. The first three parts are associated with the points in $P\cap I_{i-1}$. Since $k$ is even and $p_k$ is the first position in the interval $I_{i}$, its associated set is $P_k=I_{i-h_i-1,s_{i-h_i}+1}=I_{i-2,4}$ (and not $I_{i-1,1}$). This ensures that $\left|{P_{k-1}}\right| = \left|{P_k}\right|$. The positions before  $P_{k-1}=I_{i-2,3}$ as well as $p_{k-3}$, $p_{k-2}$, $p_{k-1}$ have already been exposed. Invoking \Cref{lemma:continuingpath} will additionally expose $P_{k-1}$, $P_k$, $p_k$ and $p_{k+1}$.}
\end{figure}
Let $S_0=\{p_1\}$, and for every even $k\in [L]$, let $S_k=S_{k-2}\cup [\max P_k] \cup \left\{p_k, p_{k+1}\right\}$. Then, iteratively for every even $k\in [L]$, 
apply \Cref{lemma:continuingpath}  with $(p_1, p_2,p_3)\gets (p_{k-1}, p_{k}, p_{k+1})$, $(P_1,P_2) \gets \left(P_{k-1},P_{k}\right)$, $l \gets |P_{k-1}|=|P_k|$, $S \gets S_{k-2}$.
Observe that indeed $(P_{k-1}\cup P_k)\cap S_{k-2} = \emptyset$, $p_{k-1}\in S_{k-2}$, and $S_k=S_{k-2} \cup [\max P_k]\cup \{p_k, p_{k+1}\}$, which---together with properties (1)--(3)---makes the lemma applicable.  Thus for every even $k$ we obtain a bound of $\frac{1-\eps}{e^2 \cdot \abs{P_{k-1}} \cdot \abs{P_{k}}}$ on the probability that $p_{k-1},p_{k}$, and $p_{k+1}$ form a dependency path when exposing $S_{k}$, conditioned on already having exposed $S_{k-2}$. Thus, $P$ forms a dependency path with probability at most
\begin{align*}
&\prod_{\mathclap{k\in [L] \colon k \text{ even}}}\text{Pr}_{S_{k}}\left[p_{k-1} \sim p_{k} \sim p_{k+1} \mid \pi(S_{k-2})\right] 
 \leq \prod_{{k \in [L]\colon k \text{ even}}} \frac{1-\eps}{e^2\cdot |P_{k-1}|\cdot|P_{k}|}=\prod_{i\colon s_i\neq0} \prod_{j=1}^{s_i} \frac{\sqrt{1-\eps}}{e  \cdot\abs{I(p_{i,j})}}.
\end{align*}
 Let $\tilde I = \{i\colon s_i \neq 0\wedge p_{i,1}=p_k \text{ for some even $k$}\}$. Then for $i\in \tilde I$
\[
	\frac{l_i}{|I(p_{i,1})|} = \frac{l_i}{l_{i-h_i}} \leq (1+\alpha)^{h_i} \frac{s_{i-h_i}+1}{s_i+1} \cdot (1+\alpha) \leq (1+\alpha)^{h_i} (s_{i-h_i}+1).
\]
Since $\sum_{i\in\tilde I} h_i\leq \log_{1+\alpha} n$ and since the mapping $f: \tilde I\to \{i\colon s_i\geq 1\}$, $f(i)=i-h_i$ is an injection we can bound the product by
\[
	\prod_{i\in \tilde I } \frac{l_i}{|I(p_{i,1})|} \leq \prod_{i \in \tilde I }  (1+\alpha)^{h_i} (s_{i-h_i}+1) \leq n \prod_{i\colon s_i\geq 1} (s_i + 1) \leq n \left(\frac{2L}{\log_{1+\alpha}(n)}\right)^{\log_{1+\alpha}(n)}.
\]
Finally by definition $|I(p_{i,j})|=l_i$ whenever $i\notin \tilde I$ or $j>1$. Therefore we obtain
\begin{align*}
& \prod_{i\colon s_i\neq0} \prod_{j=1}^{s_i} \frac{\sqrt{1-\eps}}{e  \cdot\abs{I(p_{i,j})}} \leq n    \left(\frac{2L}{\log_{1+\alpha}(n)}\right)^{\log_{1+\alpha}(n)} \prod_{i\colon s_i\neq 0} \left(\frac{\sqrt{1-\eps}}{e\cdot l_i}\right)^{s_i}.
\end{align*}

\paragraph{Union Bound Over All Positions} For fixed values of $\{s_i\}$ there are at most 
\begin{align*}
&n \cdot \prod_i \binom{\alpha\cdot (1+\alpha)^i}{s_i}
\leq n \cdot \prod_{i \colon s_i \neq 0} \left(\frac{e \cdot \alpha \cdot (1+\alpha)^i}{s_i}\right)^{s_i}  
= n\cdot \prod_{i\colon s_i \neq 0} \left(\left(\frac{s_i+1}{s_i}\right) \cdot \left(\frac{e \cdot\alpha \cdot (1+\alpha)^i}{s_i+1}\right) \right)^{s_i}
\\ & \leq n^ \cdot \prod_{i\colon s_i \neq 0} e \cdot \left( e \cdot l_i \cdot (1+\alpha)^2 \right)^{s_i} 
 \leq n^{1 + \frac{1}{\log (1+\alpha)}}\cdot \prod_{i\colon s_i \neq 0} \left(e \cdot l_i \cdot \left(1+\alpha\right)^2\right)^{s_i}
\end{align*}
choices for positions $P$ with $\abs{I_i\cap (P\setminus\{p_L\})} = s_i$ (counting $n$ choices for $p_L$), using $\left(\frac{s_i + 1}{s_i}\right)^{s_i} \leq e$ in the  second inequality and $|\{i \colon s_i \neq 0\}| \leq \log_{1+\alpha} n$ in the third inequality.  

Therefore, by a union bound, the probability of having a dependency path of length $L$ for prescribed $\{s_i\}$ is at most
$n^{2 + \frac{1}{\log (1+\alpha)}}\cdot \left(\frac{2L}{\log_{1+ \alpha}n}\right)^{\log_{1+\alpha} n} \cdot \left(\sqrt{1-\eps}\cdot (1 +  \alpha)^2\right)^{L}$. 
Since we can assign values to $\{s_i\}$ in at most $\binom{L+\log_{1+\alpha } n}{\log_{1+\alpha}n}\leq \left(e \cdot \left( \frac{L}{\log_{1+\alpha} n}+1\right)\right)^{\log_{1+\alpha}n}$
 ways, it follows that the probability of a dependency path of length $L$ is at most 
\begin{equation*}\begin{aligned}
&n^{2+ \frac{1}{\log(1+\alpha)}} \cdot\left(\frac{2L}{\log_{1+\alpha} n} \cdot e \cdot \left(\frac{L}{\log_{1+\alpha} n}+1\right)\right)^{\log_{1+\alpha}n}
\cdot \left(\sqrt{1-\eps}\cdot(1+\alpha)^2\right)^L,\end{aligned}\end{equation*}
which is $\frac{1}{\poly n}$ for a constant $\alpha$ small enough depending on $\eps$, and $L=\sqrt{A} \log n$ large enough depending on $\alpha$. 
\end{proof}

We now use the bound on the dependency length from the previous result to prove \Cref{thm:MIS}. 
\begin{proof}
By \Cref{lemma:dependencyLength}, w.h.p.\ the length of any dependency path in the interval $[\beta n]$ is $O(\log n)$, and by \Cref{lemma:reductionToLength}, this thus bounds the number of rounds needed by the parallel algorithm to process this interval by $O(\log n)$.

\paragraph{Suffix $[ \beta n, n]$}
For the suffix $[ \beta n, n]$, it can be shown that w.h.p.\ after processing the first $ \beta n$ positions the maximum degree among all remaining vertices is at most $d=O(\log n)$. See e.g. Lemma 3.1 in \cite{blelloch2012greedy}. Since each possible path of length $L$ is monotonically increasing with probability $\frac{1}{L!}\leq (\frac{e}{L})^L$, the probability of having such a path in the suffix is at most $n \cdot  \left(\frac{e\cdot d}{L}\right)^L$, which is $\frac{1}{\poly n}$ for $L=\Omega(\log n)$ large enough. Finally, observe that for the parallel algorithm to take $L$ rounds there indeed must be a monotonically increasing path of length $L$. 

\end{proof}

\subsection{Proof of \texorpdfstring{\Cref{lemma:continuingpath}}{Lemma~\ref{lemma:continuingpath} and \ref{thm:MIS}}: Continuing a Dependency Path}\label{sec:continuingpath}
\begin{proof}
As we will see next, the probability of $p_1 \sim p_2 \sim p_3$ can be bounded by considering an execution of the \emph{sequential greedy MIS} algorithm on a random ordering. 

\paragraph{Connection to Sequential Algorithm}
We will work with the following sequential algorithm.
Initially, in step $t=1$, all vertices are called \emph{alive}. Then, in each step $t\in[n]$, position $t$ is exposed (if not already exposed) and vertex $\pi(t)$ processed as follows. If $\pi(t)$ is alive in step $t$, then $\pi(t)$ is added to the MIS, and $\pi(t)$ as well as all its neighbors are called \emph{dead} (i.e., not alive) for all steps $t' > t$ after $t$. If $\pi(t)$ is dead in step $t$, then we proceed to the next step. We say that a vertex \emph{dies} in step $t$ if it is alive in step $t$ and dead in step $t+1$, and say that it is dead \emph{after} $t$.

Let $t_1 = \max P_1$ and $t_2=\max P_2$. Consider the events
\begin{enumerate}[$\mathcal{E}_1$:]
\item $\pi(p_1), \pi(p_2)$ are neighbors and alive in step $t_1 + 1\leq p_1$, and  
\item $\pi(p_2),\pi(p_3)$ are neighbors and alive in step $t_2 + 1\leq p_1$, and $\pi(p_1),\pi(p_3)$ are not neighbors. 
\end{enumerate}

Observe that these two events are necessary for $p_1\sim p_2\sim p_3$. Indeed for both $\pi(p_1)$ and $\pi(p_3)$ to enter the independent set they must not share an edge and must be alive in step $t_2 +1 \leq p_1<p_3$. Furthermore by definition of a dependency path $\pi(p_2)$ dies in step $p_1$ and therefore must be alive in step $t_1 + 1$ as well.

In the following, we call an alive and unexposed (with respect to a step and a set of exposed positions) vertex \emph{active}, and let the \emph{active degree} of a vertex be its number of active neighbors.

\paragraph{Proof Sketch} By the above observation it is enough to bound the probability that during the execution of the sequential algorithm up to step $t_2$ the events $\mathcal{E}_1$ and $\mathcal{E}_2$ occur.  
First, we will investigate how the active degree of the vertex $v:=\pi(p_1)$ evolves and thereby affects the probability of $\mathcal{E}_1$ when the sequential algorithm runs through the position set $P_1$ up to step $t_1$.
The main observation is that on the one hand, if the active degree $d$ of $v$ in step $t_1$ is low, then it is unlikely that $\pi(p_2)$ is an active neighbor of $v$ (this happens with probability $\approx \frac{d}{n}$).\footnote{Note that since there are always at least $(1-\beta)n$ unexposed positions the probability that one of $d$ alive vertices is exposed in the next step is always $\Theta\left(\frac{d}{n}\right)$.}
On the other hand, if the active degree of $v$ is high in step $t_1$ (and thus during all steps $P_1$) then $v$ is likely to die because one of its active neighbors enters the MIS ($v$ stays alive with probability $\approx \left(1-\frac{d}{n}\right)^{l}$).
The combined probability (over $P_1\cup \left\{p_2\right\}$) that $v$ stays alive and one of its active neighbors is selected for $\pi(p_2)$ is $\approx \frac{d}{n} \cdot \left(1-\frac{d}{n}\right)^l$. Therefore the probability of $\mathcal{E}_1$ is maximized for $d \approx \frac{n}{l}$, yielding a value $\frac{1}{e \cdot l}$ for one additional position, and hence $\frac{1}{(e \cdot l)^2}$ for two positions, which falls just short of our desired bound.

In that seemingly bad case, however, we will argue that $\pi(p_2)$ is likely to have a low active degree, which in turn will make $\mathcal{E}_2$ improbable (by employing a similar argument for $v':=\pi(p_2)$ when running the algorithm through $P_2$ up to step $t_2$, also exposing $\pi(p_3)$). Taken together, $\mathcal{E}_1$ \emph{and} $\mathcal{E}_2$ will have low probability in all cases, i.e., for any active degree $d$ of $v$.

\paragraph{Formal Proof} 
We may assume without loss of generality that $[t_2]\setminus \left(P_1 \cup P_2\right)\subseteq S$. If not, we expose the missing positions and add them to $S$.
Let $t_0=\min P_1$ and $S_1=S \cup P_1 \cup \{p_2\}$. Suppose that $S$ is exposed and that the sequential algorithm has run up to step $t_0-1$. We then run the sequential algorithm through the steps $[t_0, t_1]\supseteq P_1$. For $i \in \{0, \dotsc, l-1\}$, let $N_i^1$ denote the set of active neighbors of $v$ in step $P_1[i+1]$, and $N_l^1$ the set of active neighbors of $v$ after step $t_1=P_1[l]$. We use $d_i$ for the corresponding degrees, let $N_i^2$ denote the set of active neighbors of (vertices in) $N_i^1$ without $N_i^1$, and set $E_i:=E(N_i^1, N_i^2)$ and $e_i:=|E_i|$, in the respective step.  

We now analyze the probability of $\mathcal{E}_1$ by distinguishing several cases based on the set $\Omega$ of all vertex sequences $\omega \colon P_1 \rightarrow V \setminus \pi(S)$ that can be encountered when exposing $P_1$. We restrict our attention to $\omega \in \Omega^*:= \Omega \setminus \Omega_0$ for $\Omega_0:=\{\omega \in \Omega \colon v \text{ dead in steps after } t_1 \text{ under } \omega\}$, as otherwise $\mathcal{E}_1$ is impossible.  
We introduce the following auxiliary algorithm. Let $\mathcal{A}$ be the algorithm that works exactly as the sequential greedy algorithm, except that it picks a vertex for $P_1[i]$ uniformly at random from $V \setminus \left(\pi(S) \cup N_{i-1}^1\right)$ instead of from $V \setminus \pi(S)$. For $\omega \notin \Omega_0$ this set never becomes empty. If $V \setminus \left(\pi(S) \cup N_{i-1}^1\right)=\emptyset$ at some step then define $\mathcal A$ to pick vertices in some arbitrary way.
We use $\text{Pr}_{\mathcal{A}}[\omega \mid \pi(S)]$ to denote the probability of the algorithm $\mathcal{A}$ picking the sequence $\omega\in\Omega^*$ when exposing $P_1$, conditioned on $\pi(S)$.  Note that $\text{Pr}_{\mathcal{A}}[\omega\mid \pi(S)] = \prod_{i=0}^{l-1}\left(\frac{1}{n-\abs{S}-i-d_i}\right)$ and therefore
\[
\text{Pr}_{S_1}[\pi(P_1)=\omega(P_1)\mid \pi(S)] = \prod_{i=0}^{l-1}\left(\frac{1}{n-\abs{S}-i}\right) =  \text{Pr}_{\mathcal{A}}[\omega\mid \pi(S)]\cdot \prod_{i=0}^{l-1}\left(1- \frac{d_i}{n-|S|-i}\right).
\]
Moreover, the probability of having an active neighbor of $v$ at $p_2$ under $\omega$ is at most $\frac{d_l}{n-|S|-l}$. Thus, 
\begin{equation}
\begin{aligned}
\label{eq:probability}&\text{Pr}_{S_1}[\mathcal{E}_1  \text{ and }  \pi(P_1)=\omega(P_1) \mid \pi(S)]
\leq \frac{d_l}{n-|S|-l} \cdot \text{Pr}_{\mathcal{A}}[\omega \mid \pi(S)] \cdot \exp\left(-\sum_{i=0}^{l-1}\frac{d_i}{n}\right).
\end{aligned}
\end{equation}
Since the $d_i$ are decreasing (in $i$) this term is maximized for $d_0=\dots=d_l=\frac{n}{l}$, yielding an upper bound of 
\begin{equation*}\begin{aligned}
&\text{Pr}_{\mathcal{A}}[\omega \mid \pi(S)] \cdot \frac{n}{e \cdot l \cdot (n-\abs S-l)}
\leq \text{Pr}_{\mathcal{A}}[\omega \mid \pi(S)] \cdot \frac{1}{e \cdot l \cdot (1-\beta)}.
\end{aligned}\end{equation*}
Summing up over $\omega\in \Omega^*$ yields $\text{Pr}_{S_1}[\mathcal{E}_1\mid \pi(S)]\leq \frac{1}{e\cdot l\cdot (1-\beta)}$. Note that the same upper bound of $\frac{1}{e\cdot l\cdot (1-\beta)}$ can be obtained for $\Pr_{S_2}[\mathcal{E}_2\mid \pi(S_1)]$, where $S_2=S_1 \cup P_2 \cup \{p_3\}$, by repeating the argument while exposing $P_2$.

Thus we obtain our first bound of  $\left(\frac{1}{e\cdot l\cdot (1-\beta)}\right)^2$.
%\begin{equation*}
%\begin{aligned}
%&\text{Pr}_{S_2}[ \mathcal{E}_1 \text{ and } \mathcal {E}_2 \text{ and } \pi(P_1)=\omega(P_1)\mid \pi(S)] 
%\\ &=\text{Pr}_{S_1}[\mathcal{E}_1\text{ and }   \pi(P_1)=\omega(P_1) \mid \pi(S)] \cdot\text{Pr}_{S_2}[\mathcal{E}_2\mid \pi(P_1)=\omega(P_1) \text{ and } \pi(S_1)]
%\\ &\leq
%\end{aligned}
%\end{equation*}
The next step is to improve it to the claimed bound of $\frac{1-\eps}{\left(e\cdot l\right)^2}$. To that end, we distinguish three different cases for $\omega$, mainly based on the active degree $d_{\lfloor\delta l\rfloor}$ of $v$ after $\lfloor\delta l\rfloor$ steps under $\omega$, for $\delta = 0.01$, say. Note that our assumption of $l\geq 10000$ ensures that $\lfloor\delta l\rfloor \geq 0.99\delta l$.

\paragraph{Deviation from Degree $\approx \frac{n}{l}$: $\Omega_{1}=\{ \omega \in \Omega^* \colon d_{\lfloor\delta l\rfloor} > 1.1 \frac{n}{l} \text{ or } d_l < 0.8\frac{n}{l} \text{ under } \omega\}$}

\noindent Easy calculations show that when the degree deviates from the optimizer $\frac{n}{l}$, then the upper bound on \Cref{eq:probability} can be improved by a small constant factor $(1-\eps_1)$.

\paragraph{Degree $\approx \frac{n}{l}$ and Few Edges:  $\Omega_2=\{ \omega \in \Omega^* \colon d_{\lfloor\delta l\rfloor } \leq 1.1\frac{n}{l},  d_l \geq 0.8\frac{n}{l}, e_l \leq 0.6 \frac{n^2}{l^2}  \text{ under } \omega \}$}

\noindent We will argue that because $e_l$ is small, the active degree of $\pi(p_2)$ in steps $> t_1$ is likely to be low. In that case, $\mathcal{E}_2$ will have a low probability by an argument analogous to the one in the previous case where the active degree of $v$ deviates from $\frac{n}{l}$.

There can be at most $0.9375 d_{l}$ vertices in $N_l^1$ with degree into $N_l^2$ larger than $0.8 \frac{n}{l}$ after step $t_1$. All other at least $0.0625 d_{l}$ many vertices in $N_l^1$ thus have degree into $N_l^2$ at most $0.8\frac{n}{l}$. For $\pi(p_2)$ such that this is the case, analogously\footnote{The main observations needed for adapting the proof is that we can ignore all sequences for $P_2$ under which a vertex in $N_l^1$ or $N_l^2$ is exposed, since then either $v$ or $v'$ would die, and that $\pi(p_3) \in N_l^2\setminus N_l^1$ is necessary for $\mathcal{E}_2$.} to the case $\omega \in \Omega_1$ from before, we improve upon the trivial bound by a constant factor $(1-\eps_1)$. For $\pi(p_2)$ with larger degree we obtain the trivial bound of $\text{Pr}_{\mathcal{A}'}[\omega' \mid \pi(S_1)]\cdot  \frac{1}{e \cdot l\cdot(1-\beta)}$. Thus, on average we improve by some factor $(1-\eps_2)$.

\paragraph{Degree $\approx \frac{n}{l}$ and Many Edges:  $\Omega_3=\{  \omega \in \Omega^* \colon d_{\lfloor\delta l\rfloor } \leq 1.1\frac{n}{l}, d_l\geq 0.8\frac{n}{l},e_l > 0.6 \frac{n^2}{l^2}  \text{ under } \omega\}$}

\noindent
We will argue that $\sum_{\omega \in \Omega_3} \Pr_{\mathcal{A}}[\omega\mid \pi(S)]$ is bounded away from $1$ by a constant. Intuitively speaking, the more edges there are, the likelier it is that vertices in $N_i^1$ die. It is thus not too likely to have only a small drop in the active degree over all $(1-\delta) l$ steps if in every step the number of edges in $E_i$ is large. 

Run $\mathcal A$ for $\lfloor \delta l\rfloor$ steps and suppose that $d_{\lfloor\delta l\rfloor}\leq 1.1\frac{n}{l}$. We will bound the probability that $d_l \geq 0.8 \frac{n}{l}$ and $e_l> 0.6 \frac{n^2}{l^2}$ if we continue to run $\mathcal A$.

If $e_i \geq 0.6 \frac{n^2}{l^2}$ then let $M_i \subseteq E_i$ be a subset of exactly $\lceil 0.6 \frac{n^2}{l^2}\rceil$ edges, and let $X_{i+1}$ denote the number of vertices in $N_{i}^1$ which are connected through $M_i$ with the vertex selected in step $i+1$ scaled down by $\alpha\coloneqq \left(n-\abs{S\cup N^1_i}-i\right)/n$. To shorten notation, we use $\xi_i:=\pi\left(S\cup P_1\left(\left[i\right]\right)\right)$.
Then, for (say) $\beta\leq 0.01$,
\[ \E_{\mathcal{A}}\left[X_{i+1}\mid \xi_i\right] =\frac{\alpha |M_i|}{n-|S\cup N_i^1|-i} = \frac{\abs{M_i}}{n}  =\frac{\lceil0.6\frac{n^2}{l^2}\rceil}{n}. \]
If $e_i\leq 0.6 \frac{n^2}{l^2}$ then define $X_i$ to be a constant random variable (with the same expectation). Note that these random variables are uncorrelated. Indeed since $X_i$ is fully determined by $\xi_i$ while (for $j>i$) $\E_{\mathcal A}[X_j\mid \xi_i]$ does not depend on $\xi_i$ we have:
\[
\E_{\mathcal A}[X_i X_j] = \sum_{\xi_i} \E_{\mathcal A}[X_i X_j\mid \xi_i]\text{Pr}_\mathcal {A}[\xi_i] = \sum_{\xi_i} \E_{\mathcal A}[X_i\mid \xi_i]\E_{\mathcal A}[X_j\mid \xi_i]\text{Pr}_\mathcal {A}[\xi_i] =\E_{\mathcal A}[X_i] \E_{\mathcal A}[X_j].
\]

Now let $X:=\sum_{i=\lfloor\delta l\rfloor}^l X_i$ (note that if $e_l>0.6\frac{n^2}{l^2}$ then this is a lower bound on $d_{\lfloor \delta l\rfloor}-d_l$) with expectation $\E_{\mathcal{A}}[X \mid \xi_{\lfloor\delta l\rfloor}]\geq 0.5 \frac{n}{l}$ and variance
\[
\Var_{\mathcal{A}}[X_{i+1}\mid \xi_{\lfloor\delta l\rfloor}] \leq \E_{\mathcal{A}}[X_{i+1}^2\mid \xi_{\lfloor\delta l\rfloor}]\leq  d_{\lfloor\delta l\rfloor}  \cdot \E_{\mathcal{A}}[X_{i+1}\mid \xi_{\lfloor\delta l\rfloor}]
\leq \frac{n^2}{l^3},
\]
where we have used $X_{i+1}\leq d_{\lfloor \delta l \rfloor} \leq 1.1\frac{n}{l}$.

Since the $X_i$ are uncorrelated the variance of the sum is at most $\Var_{\mathcal{A}}[X\mid \xi_{\lfloor\delta l\rfloor}] \leq \frac{n^2}{l^2}$. 
Hence, by \emph{Cantelli}'s inequality (the one-sided version of \emph{Chebyshev}'s inequality),
\begin{align*}
\text{Pr}_{\mathcal{A}}\left[X\geq 0.4 \frac{n}{l}\mid \xi_{\lfloor\delta l\rfloor}\right] &\geq \Pr\left[X-\E_{\mathcal{A}}[X\mid \xi_{\lfloor\delta l\rfloor}] \geq -0.1 \frac{n}{l}\mid \xi_{\lfloor\delta l\rfloor}\right]  \\
 &\geq 1-\frac{\Var_{\mathcal{A}}[X\mid \xi_{\lfloor\delta l\rfloor}]}{0.001\frac{n^2}{l^2} +\Var_{\mathcal{A}}[X\mid \xi_{\lfloor\delta l\rfloor}]} \geq C > 0,
\end{align*}
for some constant $C$.

Thus, if we have $d_{\lfloor\delta l\rfloor}\leq 1.1\frac{n}{l}$ then with probability at least $C$ either $e_l \leq 0.6 \frac{n^2}{l^2}$ or the active degree drops by at least $0.4 \frac{n}{l}$. Both cases are excluded from $\Omega_3$. Therfore
$\sum_{\omega \in \Omega_3} \text{Pr}_{\mathcal{A}}[\omega \mid \pi(S)] \leq 1-C.$

\paragraph{Wrap-Up}
Combining these three cases, we obtain
\begin{align*}
&\sum_{\omega\in \Omega} \text{Pr}_{S_2}[ \mathcal{E}_1 \text{ and } \mathcal {E}_2 \text{ and } \pi(P_1)=\omega(P_1)\mid \pi(S)] 
\\ & \leq \sum_{\omega\in\Omega_1\cup\Omega_2} \text{Pr}_{\mathcal A}[\omega \mid \pi(S)] \cdot \frac{1-\min\left\{\eps_1 ,\eps_2\right\}}{\left((1-\beta)\cdot e\cdot l\right)^2} 
+  \sum_{\omega\in\Omega_3} \text{Pr}_{\mathcal A}[\omega \mid \pi(S)] \cdot \frac{1}{((1-\beta)\cdot e\cdot l)^2}.
\end{align*}
which is at most $\frac{1-\eps}{(e\cdot l)^2}$ for some absolute small constants $\eps,\beta>0$, since $\sum_{\omega\in\Omega_3}\text{Pr}_{\mathcal{A}}[\omega \mid \pi(S)] \leq 1-C$.
\end{proof}

\section{Corollaries}\label{appendix:implications}
Due to the well-known reductions of maximal matching and $(\Delta+1)$-vertex-coloring to MIS, our analysis in \Cref{thm:MIS} directly applies to the parallel/distributed randomized greedy algorithms for maximal matching and vertex coloring. For a review of the state of the art for these problems, we refer to \cite{DBLP:journals/jacm/BarenboimEPS16}.

\subsection{Randomized Greedy Maximal Matching}
The \emph{parallel/distributed randomized greedy maximal matching} algorithm works as follows: A random order of the edges is chosen. Then, in each round, all locally minimal edges are removed from the graph along with all their incident edges. 
\begin{corollary}\label{cor:MM}
The parallel/distributed randomized greedy maximal matching algorithm has round complexity $O(\log n)$ with high probability on graphs with $n$ vertices.
\end{corollary}
\begin{proof}
For a graph $G=(V,E)$, the line graph $L=(E,F)$ is defined to be a graph with vertex set $E$ and an edge  $\{e,e'\}\in F$ iff $e \cap e' \neq \emptyset$. Running the randomized greedy MIS algorithm on the line graph $L$ corresponds to running the randomized greedy maximal matching algorithm on $G$.
\end{proof}

\subsection{Distributed Randomized Greedy \texorpdfstring{$(\Delta+1)$}{(Delta+1)}-Vertex-Coloring} The \emph{parallel/distributed randomized greedy $(\Delta+1)$-vertex-coloring} algorithm on a graph $G=(V,E)$ with maximum degree $\Delta$ works as follows: A random order of the vertex-color pairs $V \times [\Delta + 1]$ is chosen. Then, in each round, all locally minimal pairs $(v,c)$ are removed along with all $(v',c')$ such that either $v'=v$ or $\{v,v'\}\in E$ and $c'=c$. Vertex $v$ is assigned color $c$.

\begin{corollary}\label{cor:Col}
The parallel/distributed randomized greedy $(\Delta+1)$-vertex-coloring algorithm, as defined above,\footnote{This is \emph{not} the greedy coloring algorithm where the largest available color is picked greedily.} has round complexity $O(\log n)$ with high probability on graphs with $n$ vertices and maximum degree $\Delta$. 
\end{corollary}
\begin{proof}
Luby \cite{luby1985simple, linial1987LOCAL}  presented the following reduction from $(\Delta+1)$-vertex-coloring in a graph $G$ to MIS in a graph $H$: To construct $H$, take $\Delta + 1$ copies of $G$ and add a clique among all $\Delta+1$ copies of the same vertex, for all vertices in $G$. It is easy to observe that a MIS in $H$ corresponds to a proper $(\Delta+1)$-vertex-coloring of $G$, when we assign vertex $v$ the color $i$ iff the $i^{\text{th}}$ copy of $v$ is in the MIS. Indeed, due to maximality, every vertex in $G$ is assigned at least one color, and because of the added cliques and the independence of the MIS, at most one. Moreover, having a copy of $G$ for every color guarantees that all the edges must be proper (due to independence).
\end{proof}

\subsection{Correlation Clustering}
Correlation clustering has the goal to partition nodes into clusters so that the number of miss-classified edges---that is, edges with its two endpoints in two different clusters or non-edges with endpoints in the same cluster---is minimized.
More formally, we are given a complete graph on $n$ nodes where each edge is either labeled $+$ or $-$, indicating that the corresponding nodes should be in the same or in different clusters, respectively. The goal is to group the nodes into (an arbitrary number of) clusters so that the number of $-$ edges within clusters and $+$ edges crossing clusters is minimized \cite{bansal2004correlation}.
Ailon, Charikar, and Newman \cite{ailon2008aggregating} showed that the greedy MIS algorithm, called \emph{CC-Pivot} in their paper, provides a 3-approximation for correlation clustering, when each non-MIS node is clustered with its inhibitor, that is, its lowest-rank neighbor in the MIS. Moreover, \cite{chierichetti2014correlation} argues how an iteration of this (or a similar) algorithm can be implemented in $O(1)$ rounds of \emph{MapReduce} or in $O(1)$ passes of the streaming model.

\begin{corollary}
A 3-approximation for correlation clustering can be computed in $O(\log n)$ rounds in the \emph{PRAM}, \emph{LOCAL}, and \emph{MapReduce} model, and in $O(\log n)$ passes in the streaming model.  
\end{corollary}

%\vspace{0.8cm}
\section*{Acknowledgment}
The authors want to thank Mohsen Ghaffari for suggesting this problem, providing the construction of the lower bound, and valuable comments.
%Mohsen Ghaffari would also like to acknowledge helpful conversations with George Giakkoupis, Bernhard Haeupler, Fabian Kuhn, and Julian Shun over the years.

\newpage
\bibliographystyle{alpha}
\bibliography{ref}

\appendix
\newpage 

\section{Lower Bound}\label{appendix:lower}
\begin{lemma}\label{thm:MISlower}
There exists an $n$-node graph on which the parallel randomized greedy MIS algorithm with high probability takes $\Omega(\log n)$ rounds. 
\end{lemma}
\begin{proof}
Consider a graph consisting of $\sqrt{n}$ connected components, each connected component made of $l+1$ layers for $l:=\frac{\log n}{5}$ as follows. The $i^{\text{th}}$ layer for $i \in [0,l]$ is a clique on $2^i$ vertices, and there is a full bipartite graph between any two consecutive layers. The remaining vertices not part of a layer are just isolate.

We call a path  of length $l$ strictly increasing if the $i^{\text{th}}$ vertex on the path for $0 \leq i\leq l$ is in layer $l-i$ and the of the path vertices are sorted (in the random order).
We will argue that the probability that a connected component contains such a strictly increasing path is at least $n^{-0.05}$. 

Consider the layers $U_0, \dotsc, U_{l}$ of one connected component $U$ and a random order. The probability that $U_l$ contains the minimum among $\bigcup_{i=0}^l U_i$ is at least $\frac{|U_l|}{|U|}\geq \frac{1}{4}$. Then, conditioned on the previous event, the probability of $U_{l-1}$ containing the minimum among $\bigcup_{i=0}^{l-1} U_i$ is at least $\frac{|U_{l-1}|}{|\bigcup_{i=0}^{l-1} U_i|}\geq \frac{1}{4}$. Continuing this argument, and combining the conditional probabilities, we get a lower bound of $\left(\frac{1}{4}\right)^l =n^{-0.05}$ on the probability that there is a strictly increasing path in $U$. 

Since all the $\sqrt{n}$ connected components are independent, the probability of no such component containing a strictly increasing path is at most $\left(1-n^{-0.05}\right)^{\sqrt{n}}\ll 1/\poly n$. 
Thus, with high probability, the considered graph contains such a path. 

Finally, observe that the parallel/distributed randomized greedy MIS algorithm will take at least $(l+1)/2$ rounds until it has processed such a strictly increasing path, since the algorithm processes only 2 layers in each round.
\end{proof}
\end{document}